\documentclass[12pt]{article}

\usepackage{latexsym}
\usepackage{epsfig,graphics}
\usepackage{amsmath,amssymb,amsthm}
\usepackage{graphics,epsfig,calc}
\usepackage{upref}
\newcommand{\beqn}{\begin{eqnarray}}
\newcommand{\eeqn}{\end{eqnarray}}
\newcommand{\be}{\begin{equation}}
\newcommand{\ee}{\end{equation}}
\newcommand{\ba}{\begin{array}}
\newcommand{\ea}{\end{array}}

\newcommand{\pa}{\partial}
\newcommand{\re}{\ref}
\newcommand{\ci}{\cite}
\newcommand{\la}{\label}
\newcommand{\fr}{\frac}

\newcommand{\si}{\sigma}
\newcommand{\al}{\alpha}
\newcommand{\ds}{\displaystyle}
\newcommand{\ve}{\varepsilon}
\newcommand{\de}{\delta}
\newcommand{\om}{\omega}
\newcommand{\Om}{\Omega}
\newcommand{\lam}{\lambda}

\newcommand{\De}{\Delta}

\newcommand{\tr}{\mathop{\rm tr}\nolimits}
\newcommand{\supp}{\mathop{\rm supp}\nolimits}

\newcommand{\cH}{\cal H}
\newcommand\C{{\mathbb C}}
\newcommand\R{{\mathbb R}}

\def\Re {{\rm Re\, }}
\def\Im {{\rm Im\,}}

\begin{document}

\renewcommand{\theequation}{\thesection.\arabic{equation}}
\newtheorem{theorem}{Theorem}[section]
\renewcommand{\thetheorem}{\arabic{section}.\arabic{theorem}}
\newtheorem{definition}[theorem]{Definition}
\newtheorem{deflem}[theorem]{Definition and Lemma}
\newtheorem{lemma}[theorem]{Lemma}
\newtheorem{remark}[theorem]{Remark}
\newtheorem{cor}[theorem]{Corollary}
\newtheorem{pro}[theorem]{Proposition}

\begin{titlepage}
\begin{center}

{\Large\bf
On convergence to equilibrium distribution \bigskip\\
for Dirac  equation
}
\end{center}
\vspace{2cm}
 \begin{center}
{\large A. Komech}
{\footnote{Supported partly by
Alexander von Humboldt Research Award.}$^{,2}$}\\
{\it Faculty of Mathematics of Vienna University,
1090 Vienna, Austria
\\
and IITP RAS,
Moscow, B.Karetny, 19
}\\
 e-mail:~alexander.komech@mat.univie.ac.at
\bigskip\\
{\large E. Kopylova}
{\footnote{Supported partly by
Austrian Science Fund (FWF): P22198-N13
and RFBR grant 10-01-00578-a.
}
}
\\
{\it IITP RAS,
Moscow, B.Karetny, 19}\\
 e-mail:~elena.kopylova@univie.ac.at
\end{center}
 \vspace{1cm}

\begin{abstract}
We consider  the Dirac equation in $\R^3$
with a potential, and  study the distribution $\mu_t$
of the random solution at  time $t\in\R$.
The initial measure $\mu_0$ has zero mean,
a translation-invariant covariance, and a finite mean charge
density. We also assume that $\mu_0$ satisfies a
mixing condition of Rosenblatt- or Ibragimov-Linnik-type.
The main result is the long time convergence of
projection of $\mu_t$ onto the continuous spectral space.
The limiting measure is Gaussian.

{\it Key words and phrases:}
Dirac equation,  random initial data,  mixing condition,
Gaussian measures, covariance matrices,
characteristic functional, scattering theory.
\\
{\em 2010 Mathematics Subject Classification}: 35Q41, 47A40, 60F05.
\end{abstract}
\end{titlepage}

 \section{Introduction}
This paper can be considered as a continuation of our
papers \ci{DKKS}-\ci{DKM}, \ci{KKM}  which concern the long time
convergence to equilibrium distribution for the linear wave,
Klein-Gordon and Schr\"odinger equations.

The convergence should clarify
the distinguished role of the canonical Maxwell-Boltzmann-Gibbs
equilibrium distribution in statistical physics.
One of fundamental examples is the Kirchhoff-Planck
black body radiation law which specify the equilibrium distribution
for the Maxwell equations, and served as a basis for
creation of quantum mechanics.
The law likely should be correlation function of limiting
equilibrium measure for coupled Maxwell-Schr\"odinger or
Maxwell-Dirac equations.

Our ultimate goal would be the proof of the  convergence for
nonlinear hyperbolic PDEs. At the moment,
a unique result in nonlinear case has been proved by Jaksic and Pillet
 for  wave equation coupled to a nonlinear
finite dimensional Hamiltonian system \ci{JP}.

The main peculiarity of the problem
is the time-reversibility of dynamical equations.
For infinite particle systems this
difficulty was discussed in
Boltzmann-Zermelo debates (1896-1897).
Many attempts were made to deduce the convergence from
an ergodicity for such systems by H. Poincar\'e,
 G. Birkhoff, A. Hinchin, and many others.
However, the ergodicity is not proved until now.

In 1980  R. Dobrushin and Yu. Suhov have introduced
a totally new idea for obtaining the
convergence to equilibrium measures imposing a mixing condition
on initial distributions \ci{DS}
in the context of  infinite particle systems.

We develop this approach for hyperbolic PDEs.
In \ci{DKKS}-\ci{DKM}, \ci{KKM1}-\ci{KKM} the
convergence  to equilibrium distributions has been proved  
for the linear wave,
Klein-Gordon and Schr\"odinger equations with potentials,
for the harmonic crystal, and for the free Dirac equation.
The initial distribution are translation invariant
and satisfy the  mixing condition
of Rosenblatt or Ibragimov-Linnik type.

Here
we consider the linear Dirac equation with the Maxwell potentials in $\R^3$:
\beqn\la{1}
\left\{\ba{l}
i\dot\psi(x,t)=H\psi(x,t):=[-i\al\cdot\nabla+\beta m+V(x)]\psi(x,t)\\
\psi(x,0)=\psi_0(x)
\ea
\right|~~~x\in\R^3
\eeqn
where $\psi(x,t)\in\C^4$, $m>0$ and  $\al=(\al_1,\al_2,\al_3)$.
The hermitian matrices $\beta=\al_0$ and $\al_k$ 
satisfy the following  relations:
$$
\left\{
\ba{ll}
\al^*_k=\al_k,\\
\al_k\al_l+\al_l\al_k=2\delta_{kl}I
\ea
\right|\quad k,l=0,1,2,3,4.
$$
The standard form of the Dirac matrices
$\al_k$ and  $\beta$
(in $2\times 2$ blocks) is
\be\la{ba}
\beta= \left(
\ba{ll}
I & 0\\
0 & -I\\
\ea  \right),\quad
\al_k=
\left(
\ba{ll}
0 & \sigma_k\\
\sigma_k & 0\\
\ea  \right)\quad (k=1,2,3),
\ee
where  $I$ denotes the  unit matrix, and
\be\la{sigma}
\sigma_1=  \left(
\ba{ll}
0 & 1\\
1 & 0\\
\ea         \right),\quad
\sigma_2=   \left(
\ba{ll}
0 & -i\\
i & 0\\
\ea         \right),\quad
\sigma_3=   \left(
\ba{ll}
1 & 0\\
0 & -1\\
\ea         \right).
\ee
We assume the following conditions:
\smallskip\\
{\bf E1.}
The potential $V\in C^{\infty}(\R^3)$  is a hermitian
$4\times 4$ matrix function such that
\be\la{V}
|\pa^{\al}V(x)|\le C(\al)\langle x\rangle^{-\rho-|\al|},\quad
\langle x\rangle^\si=(1+|x|^2)^{\si/2}
\ee
with some $\rho>5$.
\smallskip\\
{\bf E2.} The operator $H$ presents neither resonance
nor eigenvalue at thresholds.
\smallskip\\
Under the condition ${\bf E2}$ the operator $H$ has a finite set of
eigenvalues $\om_j\in (-m,m)$, $j=1,...,N$ with the corresponding
eigenfunctions $\zeta_j^1,...,\zeta_j^{k_j}$, where $k_j$ is the
multiplicity of $\om_j$. Denote by $P_j$ the Riesz projection onto
the corresponding eigenspaces and by
\be\la{pc}
P_c:=1-P_d,\quad P_d=\sum\limits_jP_j
\ee
the projections onto the continuous and discrete
spectral spaces of  $H$.

We fix an arbitrary  $\delta>0$ such that $5+\delta<\rho$ and
consider the solutions $\psi(x,t)\in \C^4$ with initial
data $\psi_0(x)$ which
are supposed to be a random element of the weighted Sobolev space
${\cal H}=L^2_{-5/2-\delta}$,
see Definition \ref{space} below.
The distribution of $\psi_0$ is a Borel probability
measure $\mu_0$ on  $\cal H$ with zero mean
satisfying some additional assumptions, see Conditions {\bf S1-S3} below.
Denote by  $\mu_t$, $t\in\R$,
the measure on  $\cal H$, giving the  distribution of the random
solution  $\psi(t)$ to problem (\re{1}).
We identify the complex and real spaces $\C^4\equiv \R^8$, and
$\otimes$ stands for the tensor product of real vectors.
The correlation functions of the initial measure
are supposed to be translation-invariant:
\be\la{1.9'}
  Q_0(x,y):= E\Big(\psi_0(x)\otimes\psi_0(y)\Big)=
  q_0(x-y),\,\,\,x,y\in\R^3.
\ee
We also assume that the initial mean charge density  is finite:
\be\la{med}
  e_0:=E \vert \psi_0(x)\vert^2=
  \tr q_0(0)<\infty,\quad x\in\R^3.
\ee
Finally, we assume that the measure $\mu_0$ satisfies a mixing
condition of a Rosenblatt- or Ibragimov-Linnik type, which means that
\be\la{mix}
  \psi_0(x)\,\,\,\,   and \, \, \,\,\psi_0(y)
  \,\,\,\,  are\,\,\,\, asymptotically\,\,\,\, independent\,\, \,\,
  as \,\, \,\,|x-y|\to\infty.
\ee
Let $P^*_c\mu_t$ denote the projection of $\mu_t$ onto
the space ${\cal H}_c:=P_c\cal H$.
Our main result is the (weak) convergence
of  $P^*_c\mu_t$ to a limiting measure $\nu_\infty$,
\be\la{1.8i}
  P^*_c\mu_t \rightharpoondown\nu_\infty,\,\,\,\, t\to \infty,
\ee
which is an equilibrium Gaussian measure on ${\cal H}_c$. 
A similar convergence
holds for $t\to-\infty$ since our system is time-reversible.
\medskip

The convergence (\re{1.8i}) for the free Dirac equation with
$V(x)\equiv 0$ has been proved in \ci{DKM}.
The case of the perturbed Dirac equation with $V\not =0$ requires
new constructions due to
the absence an explicit formula for the solution.
To reduce the case of perturbed equation to the case of free
equation  we formally need a scattering theory for the solutions of
infinite global charge.
We manage a dual scattering theory for finite charge solutions to avoid
the infinite charge scattering theory:
\be\la{dsti}
  P_cU'(t)\phi=U'_0(t)W\phi+r(t)\phi,\quad t\ge 0.
\ee
Here  $U'_0(t)$ and $U'(t)$ are  a 'formal adjoint' to the dynamical groups
$U_0(t)$ and $U(t)$ of the free equation with $V\equiv 0$ and
equation (\re{1}) with $V\not =0$ respectively.
The remainder $r(t)$ is small in the mean:
\be\la{rem}
  E|\langle\psi_0,r(t)\phi\rangle|^2\to 0,\,\,\,t\to\infty.
\ee
where $\langle\cdot,\cdot\rangle$ is defined in (\re{1.5'}).
This version of scattering theory
is based on the weighted energy decay established in \ci{Bo}.

 \section{Main results}
\subsection{Well posedness}
\begin{definition}\la{space}
For $s,\si\in\R$, let us denote by $H^s_\si=H^s_\si (\R^3,\C^4)$
the weighted Sobolev spaces with the finite norms
$$
  \Vert\psi\Vert_{H^s_\si}=\Vert\langle x
  \rangle^\si\langle\nabla\rangle^s\psi\Vert_{L^2}<\infty.
$$
\end{definition}
We set $L^2_\si=H^0_\si$.
Note, that the multiplication by $V(x)$
is bounded operator $L^2_\si \to L^2_{\si+\rho}$.
The finite speed of propagation for equation (\ref{1}) implies
\begin{pro}\la{p1.1}
i) For any $\psi_0 \in L^2_{-\si}$ with $0\le\si\le \rho$
 there exists  a unique solution $\psi(\cdot,t)\in C(\R,\,L^2_{-\si})$
 to the Cauchy problem (\re{1}).\\
ii) For any  $t\in \R$, the
operator $U(t):\psi_0\mapsto  \psi(\cdot,t)$
 is continuous in $L^2_{-\si}$.
\end{pro}
\begin{proof}
Fist, consider the free Dirac equation:
\beqn\la{fD}
  \left\{\ba{l}
  \dot\chi(x,t)=H_0\chi(x,t)=
  (-\alpha\cdot\nabla-i\beta m)\chi(x,t)\,\,\,\,x\in\R^3,\\
  \chi(x,0)=\psi_0(x).
\ea \right.
\eeqn
Let $s\in\R$ and $\psi_0 \in L^2_{s}$.
In the Fourier space the solution to (\re{fD}) reads
$$
\hat\chi(k,t)=e^{i(\al\cdot k-\beta m)t}\hat\psi_0(k).
$$
Since $\hat\psi_0\in H^{s}$ then
$\hat\chi(\cdot,t)\in H^{s}$ and the bounds hold
\be\la{bm}
\Vert \chi(\cdot,t)\Vert_{L^2_s}=
C\Vert \hat\chi(\cdot,t)\Vert_{H^s}\le C_s(t)
\Vert \hat\psi_0\Vert_{H^s}
\le C'_s(t)\Vert \psi_0\Vert_{L^2_s}\,.
\ee
\smallskip\\
Now consider  perturbed equation (\re{1}).
Let $0\le\si\le \rho$ and $\psi_0\in L^2_{-\si}$.
We seek the solution to (\re{1}) in the form
\be\la{split}
\psi(x,t)=\chi(x,t)+\phi(x,t),
\ee
where $\chi(t)=U_0(t)\psi_0\in L^2_{-\si}$
is the solution to  free equation (\re{fD}), and
\be\la{vsp}
\dot\phi(x,t)=H\phi(x,t)+V\chi(x,t),\quad \phi(x,0)=0.
\ee
Since $\phi(0)=0$ and $V\chi\in L^2$ then
there exists the unique solution  $\phi(t)\in L^2$
to (\re{vsp}) which is given by  Duhamel representation:
$$
\phi(t)=\int\limits_0^t U(t-\tau)V\chi(\tau)d\tau.
$$
Finally, by charge conservation for the Dirac
equation we obtain
$$
\Vert U(t-\tau)V\chi(\tau)\Vert_{L^2_{-\si}}
\le \Vert U(t-\tau)V\chi(\tau)\Vert_{L^2}=\Vert V\chi(\tau)\Vert_{L^2}
\le C\Vert \chi(\tau)\Vert_{L^2_{-\rho}}
\le C\Vert \chi(\tau)\Vert_{L^2_{-\si}}<\infty\,.
$$
\end{proof}
\subsection{Random solution. Convergence to equilibrium}
Let $(\Om,\Sigma,P)$ be a probability space
with expectation $E$
and ${\cal B}(\cH)$ denote the Borel $\sigma$-algebra
in $\cH$.
We assume that $\psi_0=\psi_0(\om,\cdot)$ in (\re{1})
is a measurable random function
with values in $(\cH,\,{\cal B}(\cH))$.
In other words, $(\om,x)\mapsto \psi_0(\om,x)$
is a measurable  map
$\Om\times\R^3\to\C^4$ with respect to the
(completed) $\sigma$-algebras
$\Sigma\times{\cal B}(\R^3)$ and ${\cal B}(\C^4)$.
Then, owing to Proposition \re{p1.1},
$\psi(t)=U(t) \psi_0$ is again a measurable  random
function with values in
$(\cH,{\cal B}(\cH))$.
We denote by $\mu_0(d\psi_0)$ a Borel probability measure
in $\cH$ giving
the distribution of the random function $\psi_0$.
Without loss of generality,
we assume $(\Om,\Sigma,P)=
(\cH,{\cal B}(\cH),\mu_0)$
and $\psi_0(\om,x)=\om(x)$ for
$\mu_0(d\om)\times dx$-almost all
$(\om,x)\in {\cH}\times\R^3$.
\begin{definition}
$\mu_t$ is a probability measure on $\cH$
which gives the distribution of $\psi(t)$:
\begin{eqnarray}\la{1.6}
  \mu_t(B) = \mu_0(U(-t)B),\quad 
  \forall B\in {\cal B}({\cH}),\,\,\,   t\ge 0.
\eeqn
\end{definition}
Denote by $P^*_c\mu_t$ the projection of measure $\mu_t$
onto  ${\cal H}_c=P_c\cal H$:
\begin{eqnarray}\la{Pm}
  P^*_c\mu_t(B) = \mu_t(P_c^{-1}B),\quad 
  \forall B\in {\cal B}({{\cal H}_c}),\,\,\,   t\ge 0.
\eeqn

Our main goal is to derive
the weak convergence of  $P^*_c\mu_t$
in the Hilbert space
 $P_cH^{-\ve }_{-\si}$ for any  $\ve>0$, and $\si>5/2+\delta$:
 \be\la{1.8}
 P^*_c\mu_t\,\buildrel {\hspace{2mm} P_cH^{-\ve }_{-\si}}\over
 {- \hspace{-2mm} \rightharpoondown }
 \, \nu_\infty {\rm ~~ as ~~}t\to \infty,
 \ee
 where $\nu_\infty$ is a Borel probability measure
 on  $P_cH^{-\ve }_{-\si}$.
 By definition, this means the convergence
 \be\la{1.8'}
 \int f(\psi)P_c^*\mu_t(d\psi)\rightarrow
 \int f(\psi)\nu_\infty(d\psi){\rm ~~ as ~~}t\to \infty.
 \ee
 for any bounded and continuous functional $f(\psi)$
 in   $P_cH^{-\ve }_{-\si}$.

Set
 ${\cal R}\psi\equiv (\Re \psi,\Im \psi)=
\{\Re\psi_1,\dots,\Re\psi_4,\Im\psi_1,\dots,\Im\psi_4\}$
for $\psi= (\psi_1,\dots\psi_4)\in \C^4$
and denote by  ${\cal R}^j\psi$
the  $j$-th component of the vector
${\cal R}\psi$, $j=1,...,8$.
The brackets $(\cdot ,\cdot)$ mean
the inner product
in the real  Hilbert spaces
 $L^2\equiv L^2(\R^3)$,  in $L^2\otimes \R^N$, or in some
their  different extensions.
For $\psi(x),\phi(x)\in L^2(\R^3,\C^4)$, write
\be\la{1.5'}
\langle\psi,\phi\rangle:=
 ({\cal R}\psi,{\cal R}\phi)=
\sum\limits_{j=1}^8({\cal R}^j\psi,{\cal R}^j\phi).
\ee
\begin{definition}\la{Cor-f}
The correlation functions of the measure $\mu_0$ are
defined by
\be\la{qd}
  Q_0^{ij}(x,y)\equiv E\Big({\cal R}^i\psi_0(x){\cal R}^j\psi_0(y)\Big)
\quad\mbox{for almost all }\,\, x,y\in\R^3,~~i,j=1,...,8,
\ee
provided that the  expectations in the right-hand side are finite.
\end{definition}
Denote  by $D$ the space of complex- valued functions
in $C_0^\infty(\R^3)$ and write ${\cal D}:=[D]^4$.
For a Borel probability  measure $\mu$ 
denote by $\hat\mu$
the characteristic functional (the Fourier transform)
$$
\hat \mu(\phi )  \equiv
 \int\exp(i\langle\psi,\phi \rangle)\,\mu(d\psi),\,\,\,
 \phi\in  {\cal D}.
$$
A  measure $\mu$ is said to be Gaussian (with zero expectation) if
its characteristic functional is of the form
$$
\ds\hat{\mu}(\phi ) =  \ds \exp\{-\fr{1}{2}
 {\cal Q}(\phi , \phi )\},\,\,\,\phi \in {\cal D},
$$
where ${\cal Q}$ is a real nonnegative
quadratic form on ${\cal D}$.
A measure $\mu$ on $\cH$ is said to be
translation-invariant if
$$
\mu(T_h B)= \mu(B),\,\,\,\,\, B\in{\cal B}({\cH}),
\,\,\,\, h\in\R^3,
$$
where $T_h \psi(x)= \psi(x-h)$, $x\in\R^3$.

\subsection{Mixing condition}
Let $O(r)$ be the set of all pairs of open bounded subsets
${\cal A}, {\cal B}\subset \R^3$ at the distance not less
than $r$, dist$({\cal A},\,{\cal B})\geq r$, and let
$\sigma ({\cal A})$ be the $\sigma$-algebra in
$\cH$ generated by the linear functionals
$\psi\mapsto\, \langle\psi,\phi\rangle$,
where  $\phi\in  {\cal D}$ with $ \supp \phi \subset {\cal A}$.
Define the
Ibragimov-Linnik mixing coefficient
of a probability measure $\mu_0$ on  $\cH$
by the rule (cf. \ci[Def. 17.2.2]{IL})
\be\la{7}
\varphi(r)\equiv
\sup_{({\cal A},{\cal B})\in O(r)} \sup_{
\ba{c} A\in\si({\cal A}),B\in\si({\cal B})\\ \mu_0(B)>0\ea}
\fr{| \mu_0(A\cap B) - \mu_0(A)\mu_0(B)|}{ \mu_0(B)}.
\ee
\begin{definition}
We say that the measure $\mu_0$ satisfies the strong uniform
Ibragimov-Linnik mixing condition if
\be\la{1.11}
\varphi(r)\to 0\quad{\rm as}\quad r\to\infty.
\ee
\end{definition}
We specify the rate of  decay of $\varphi$ below
(see Condition {\bf S3}).


\subsection{Main assumptions and results}

We assume that the measure $\mu_0$
has the following properties {\bf S0--S3}:
\bigskip\\
{\bf S0}
$\mu_0$ has zero expectation value,
$$
  E\psi_0(x)  \equiv  0,\,\,\,x\in\R^3.
$$
{\bf S1}
$\mu_0$ has translation invariant correlation functions,
\be\la{corf}
  Q_0^{ij}(x,y)\equiv E\Big({\cal R}^i\psi_0(x)
  {\cal R}^j\psi_0(y)\Big)
  =q_0^{ij}(x-y),\quad i,j=1,...,8
\ee
for almost all $x,y\in\R^3$.\\
{\bf S2}   $\mu_0$ has  finite
mean charge density, i.e. Eqn (\re{med}) holds.\\
{\bf S3}
 $\mu_0$ satisfies the strong uniform
Ibragimov-Linnik mixing condition, with
\be\la{1.12}
  \int _0^\infty r^{2}\varphi^{1/2}(r)dr <\infty.
\ee
\begin{remark}
The examples of measures on  $L^2_{loc}(\R^3)$
satisfying properties {\bf S0}-{\bf S3}
have been constructed in \ci{DKKS} (see \S\S 2.6.1-2.6.2).
The measures on $L^2_{-\si}$ with any $\si>3/2$ can be construct
similarly.
\end{remark}
Introduce the following $8\times 8$
real valued matrices (in $4\times 4$ blocks)
\be\la{matr}
\Lambda_1= \left(
\ba{ll}
\alpha_1 & 0\\
0 & \alpha_1\\
\ea  \right),~~
\Lambda_2= \left(
\ba{ll}
0&i\alpha_2\\
- i\alpha_2 &0\\
\ea  \right),~~
\Lambda_3= \left(
\ba{ll}
\alpha_3 & 0\\
0 & \alpha_3\\
\ea  \right),~~
\Lambda_0= \left(
\ba{ll}
0&-\beta\\
 \beta&0\\
\ea  \right).
\ee
Note that
$\Lambda_k^T=\Lambda_k$,
$k=1,2,3$, $\Lambda_0^T=-\Lambda_0$. Write
\be\la{LP}
{\Lambda}=(\Lambda_1,\Lambda_2,\Lambda_3),\quad
 P={\Lambda}\cdot\nabla+m\Lambda_0.
\ee
For almost all $x,y\in\R^3$, introduce the  matrix-valued function
\be\la{Q}
  Q_{\infty}(x,y)\equiv
  \Big(Q_{\infty}^{ij}(x,y)\Big)_{i,j=1,\dots,8}
  =\Big(q_\infty^{ij}(x-y)\Big)_{i,j=1,\dots,8}.
\ee
Here
\be\la{qk}
\hat q_{\infty}(k)=\fr 12\hat q_0(k)-\fr 12\hat{\cal P}(k)
\hat P(k)\hat q_0(k)\hat P(k),
\ee
$\hat P(k)=-i\Lambda\cdot k+m\Lambda_0$,
$\hat{\cal P}(k)=1/(k^2+m^2)$, and $\hat q_0(k)$ is the
Fourier transform of the correlation matrix of the measure
$\mu_0$ (see \re{corf}). We formally have
\beqn \la{1.13'}
 q_\infty (z)
=\frac{1}{2}q_0(z)+\frac{1}{2}{\cal P}*P q_0(z)P
\eeqn
where ${\cal P} (z)=e^{-m|z|}/(4\pi|z|)$
is the fundamental solution for the operator $-\De+m^2$, and
$*$ stands for the convolution of distributions.
\begin{lemma}\la{qq}
Let conditions  {\bf S0}, {\bf S2} and  {\bf S3} hold. Then
\be\la{qp}
q_0\in L^p(\R^3),\quad p\ge 1.
\ee
\end{lemma}
\begin{proof}
Conditions  {\bf S0}, {\bf S2} and  {\bf S3} imply
(cf. \ci[Lemma 17.2.3]{IL}) that
$$
|q_0^{ij}(z)|\le Ce_0\varphi^{1/2}(|z|),\quad z\in\R^3,
\quad i,j=1,...,8.
$$
The mixing coefficient $\varphi$ is bounded, hence
$$
\int |q_0^{ij}(z)|^pdz\le C\int\varphi^{p/2}(|z|)dz
\le C_1\int_0^\infty r^2\varphi^{1/2}(r)dr<\infty
$$
by  (\re{1.12}).
\end{proof}
Lemma \re{qq} with $p=2$ imply that $\hat q_0\in L^2$.
Hence, $\hat q_{\infty}\in L^2$ by (\re{qk}), and $q_\infty$ also
belongs to $L^2$ by (\re{1.13'}).

Denote by
${\cal Q}_{\infty}$ a real quadratic form on $L^2$ defined by
$$
  {\cal Q}_\infty (\phi,\phi)\equiv
  (Q_\infty(x,y),{\cal R}\phi(x)\otimes{\cal R}\phi(y))
  =\sum\limits_{i,j=1}^{8}
  \int_{\R^3\times\R^3}Q_\infty^{ij}(x,y){\cal R}^i\phi(x)
  {\cal R}^j\phi(y)dxdy
$$
\begin{cor}\la{coro}
The form ${\cal Q}_{\infty}$ is continuous on $L^2$
because $\hat q_0(k)$ and then $\hat q_\infty(k)$
are bounded by Lemma \ref{qq} and formula (\re{qk}).
\end{cor}

Our main result is the following:
\begin{theorem}\la{tA}
\it    Let   $m>0$,
and let conditions {\bf E1--E2}, {\bf S0--S3} hold.
 Then \\
i) the convergence in (\re{1.8}) holds for any $\ve>0$ and $\si>5/2+\delta$.\\
ii) the limiting measure $ \mu_\infty $ is a Gaussian equilibrium
measure on ${\cal H}_c$.\\
iii) the  characteristic functional of $\nu_\infty$ is of  the form
$$
\ds\hat { \nu}_\infty (\phi ) = \exp
\{-\fr{1}{2}
{\cal  Q}_\infty (W \phi, W \phi)\},\,\,\,
\phi \in {\cal D},
$$
where  $W: {\cal D}\to L^2$ is
 a linear continuous operator.
\end{theorem}
\subsection{Remark on various mixing conditions for initial measure}
We use the {\it strong uniform}
Ibragimov-Linnik mixing condition for the simplicity of our presentation.
The {\it uniform} Rosenblatt mixing condition
\ci{Ros} with a higher degree $>2$ in the bound (\re{med}) is also  sufficient.
In this case we assume that there exists an $\epsilon$, $\epsilon >0$,
such that
$$
\sup\limits_{x\in\R^3}
E |\psi_0(x)|^{2+\epsilon}<\infty.
$$
Then condition (\re{1.12}) requires the following  modification:
$$
\int _0^\infty r\al^{p}(r)dr <\infty,\quad
p=\min(\epsilon/(2+\epsilon), 1/2),
$$
where $\al(r)$ is the  Rosenblatt mixing coefficient  defined
as in  (\re{7}), but without the denominator $\mu_0(B)$.
The statements of Theorem \re{tA} and their proofs remain essentially
unchanged.
\setcounter{equation}{0}
   \section{Free Dirac equation }

Here we consider the free Dirac equation (\re{fD})
We have
$$
(\pa_t-\al\cdot\nabla-i\beta m)(\pa_t+\al\cdot\nabla+i\beta m)
=\pa^2_t-\Delta+m^2
$$
Then the fundamental solution $G(x,t)$ of the free Dirac operator
reads
\be\la{G}
G_t(x)=(\pa_t-\al\cdot\nabla-i\beta m){\cal E}_t(x)
\ee
where ${\cal E}_t(x)$ is the fundamental solution of the
Klein-Gordon operator $\pa^2_t-\Delta+m^2$:
\be\la{cE}
{\cal E}_t(x)=F^{-1}_{k\to x}\fr{\sin\om t}{\om},\quad
\om=\om(k)=\sqrt{|k|^2+m^2}.
\ee
Using the notations (\re{matr}) and (\re{LP}), we obtain in real form
\be\la{xi-sol}
{\cal R}\chi(t)={\cal G}_t*{\cal R}\psi_0, \quad {\cal G}_t=(\pa_t-P){\cal E}_t.
\ee
The convolution exists since the distribution
${\cal E}_t(x)$ is supported by the ball $|x|\le t$.
Now we derive an explicit formula for the correlation function
\be\la{qt}
Q_t(x,y)=q_t(x-y)=E\Big({\cal R}\chi(x,t)\otimes{\cal R}\chi(y,t)\Big)
\ee
\begin{lemma}\la{cor-f} (cf. \ci[Formula (4.6)]{DKM})
The correlation function $Q_t(x,y)$ reads
\beqn\nonumber
Q_t(x,y)&=&q_t(x-y)
=F^{-1}_{k\to x-y}
\Big[\fr{1+\cos 2\om t}2\hat q_0(k)
-\fr{\sin 2\om t}{2\om}(\hat q_0(k)P(k)-P(k)\hat q_0(k))\\
\la{Qt-rep}
&-&\fr{1-\cos 2\om t}{2\om^2}P(k)\hat q_0(k)P(k)\Big]
\eeqn
\end{lemma}
\begin{proof}
Applying the Fourier transform to (\re{xi-sol}) we obtain
\be\la{hc}
\widehat{{\cal R}\chi}(k,t)=\hat {\cal G}_t(k)\widehat{{\cal R}\psi}_0(k)
=\Big(\cos\om t-\hat P(k)\fr{\sin\om t}{\om}\Big)\hat\psi_0(k)
\ee
By translation invariance condition (\re{corf}) we have
$$
E(\widehat{{\cal R}\psi}_0(k)\otimes
\widehat{{\cal R}\psi}_0(k'))=F_{x\to k,y\to k'}
q_0(x-y)=(2\pi)^3\delta(k-k')\hat q_0(k)
$$
Then (\re{hc}) implies that
$$
E(\widehat{{\cal R}\chi}(k,t)\otimes
\widehat{{\cal R}\chi}(k',t))
=(2\pi)^3\delta(k-k')\hat {\cal G}_t(k)\hat q_0(k)\hat {\cal G}_t^*(k)
$$
Therefore,
$$
\hat q_t(k)=\hat {\cal G}_t(k)\hat q_0(k)\hat {\cal G}_t^*(k)
=\Big(\cos\om t-\hat P(k)\fr{\sin\om t}{\om}\Big)\hat q_0(k)
\Big(\cos\om t+\hat P(k)\fr{\sin\om t}{\om}\Big)
$$
since $\hat P^*(k)=-\hat P(k)$. Hence (\re{Qt-rep}) follows.
\end{proof}
\begin{cor}
For any $z\in\R^3$ the convergence holds
$$
q_t(z)\to q_\infty(z),\quad t\to\infty
$$
where $q_\infty(z)$ is defined in (\re{1.13'}).
\end{cor}
\begin{proof}
The convergence follows from (\re{Qt-rep}) since the integrals
with the oscillatory functions converge to zero.
\end{proof}
Below we will need the following  lemma:
\begin{lemma}\la{p1}
Let  Conditions {\bf S0--S3} hold.
Then for any $\si>3/2$ the bound holds
\be\la{bpp}
\sup\limits_{t\ge 0} E\Vert\chi(\cdot,t)\Vert^2_{L^2_{-\si}}
<\infty
\ee
\end{lemma}
\begin{proof}
Denote
$$
e_t(x):=E|\chi(x,t)|^2,\quad x\in\R^3.
$$
The mathematical expectation is finite for almost all $x\in\R^3$
by (\re{bm}) with $s=-\si$ and the Fubini theorem.
Moreover, $e_t(x)=e_t$ for almost all $x\in\R^3$ by $\bf S1$.
Formula (\re{Qt-rep}) implies
\beqn
q_t(0)&=&\fr 1{(2\pi)^3}\int\Big[\cos^2(\om t)\hat q_0(k)
-\fr{\sin 2\om t}{2\om}(\hat q_0(k)P(k)-P(k)\hat q_0(k))\\
\nonumber
&-&\fr{\sin^2\om t}{\om^2}P(k)\hat q_0(k)P(k)\Big]dk,
\eeqn
Then $e_t=\tr q_t(0)\le Ce_0$.
Hence for $\si>3/2$ we obtain
$$
E\Vert \chi(\cdot,t)\Vert^2_{L^2_{-\si}}
=e_t\int(1+|x|^2)^{-\si}dx\le C(\nu)e_0
$$
and then (\re{bpp}) follows.
\end{proof}
We will use also the following result:
\begin{pro}\la{p2}(see \ci[Proposition 2.8]{DKM}, \ci[Proposition 3.3]{DKKS}).
Let   Conditions {\bf S0--S3} hold. Then
for any $\phi\in {\cal D}$,
\be\la{conv}
E\exp\{i\langle U_0(t)\psi_0,\phi\rangle\}
 \rightarrow
\exp\{-\fr{1}{2}{\cal Q}_\infty(\phi,\phi)\},
 \,\,\,t\to\infty.
 \ee
\end{pro}
\begin{remark}
In \ci{DKM} the phase space $L^2_{loc}(\R^3)\otimes\C^4$
has been considered. Nevertheless, all the steps of proving
the convergence (\re{conv}) in \ci{DKM}
remain true if we change  $L^2_{loc}(\R^3)\otimes\C^4$ by $L^2_{-\si}$ with
 $\si>3/2$.
\end{remark}
\setcounter{equation}{0}
\section{Perturbed Dirac equation.}
\subsection{Scattering Theory}
To deduce Theorem \re{tA}
we construct the dual scattering theory (\ref{dsti})
for finite energy solutions
using the Boussaid  results, \ci{Bo}.
\begin{lemma}\la{Bous} (see \ci[Theorem 1.1]{Bo})
Let conditions {\bf E1-E2} hold and $\si>5/2$. Then the bound holds
\be\la{full}
\Vert P_cU(t)\psi\Vert_{L^2_{-\si}}
\le C(1+|t|)^{-3/2}\Vert \psi\Vert_{L^2_{\si}},\quad t\in\R.
\ee
\end{lemma}
Note that for $\psi_0\in L^2$ the solutions $U_0(t)\psi_0$ and
$U(t)\psi_0$ to problems (\ref{fD})
and (\ref{1}),  respectively, also belong to $L^2$
and the  charge conservation holds:
\be\la{60}
\Vert U(t)\psi_0\Vert=\Vert\psi_0\Vert,~~
\Vert U_0(t)\psi_0\Vert=\Vert\psi_0\Vert.
\ee
Here and below $\Vert\cdot\Vert$ is the norm in $L^2$.

For  $t\in\R$,  introduce
the operators $U'_0(t)$ and $U'(t)$
which are  conjugate to the operators $U_0(t)$ and  $U(t)$ on $L^2$:
\be\la{def}
(\psi,U'_0(t)\phi)=
(U_0(t)\psi,\phi),\quad
(\psi,U'(t)\phi)=(U(t)\psi,\phi),\quad\psi,\phi\in L^2.
\ee
Here $(\cdot,\cdot)$ stands for the hermitian scalar product
in $L^2(\R^3,\C^4)$.
The adjoint groups admit a convenient description:
\begin{lemma}\la{8.3}
For $\phi\in L^2$ the function
$U'_0(t)\phi_0=\phi(\cdot,t)$  is the solution to
\be\la{UP0}
\dot\phi(x,t)=[\alpha\cdot\nabla+i\beta m]\phi(x,t),~~
\phi(x,0)=\phi_0(x).
\ee
\end{lemma}
\begin{proof}
Differentiating the first equation of (\ref{def}) with
$\psi,\phi\in {\cal D}$, we obtain
\be\la{UY}
(\psi,\dot U'_0(t)\phi)=(\dot U_0(t)\psi,\phi).
\ee
The group $U_0(t)$ has the generator
\be\la{A0}
  {\cal A}_0=-\alpha\cdot\nabla-i\beta m.
\ee
Therefore,
the generator of $U'_0(t)$   is the conjugate operator
\be\la{A'0}
  {\cal A}'_0=\alpha\cdot\nabla+i\beta m.
\ee
Hence, (\ref{UP0}) holds, where
$\dot \phi(t)={\cal A}'_0\phi(t).$
\end{proof}
Similarly, we obtain
\begin{lemma}\la{8.4}
For $\phi\in L^2$ the function
$U'(t)\phi=\phi(x,t)$  is the solution to
\be\la{UP}
\dot\phi(x,t)=[\alpha\cdot\nabla+i\beta m+iV]\phi(x,t),~~
\phi(x,0)=\phi(x).
\ee
\end{lemma}
\begin{cor}\la{co7.1}
i) $U'_0(t)=U_0(-t)$, $U'(t)=U(-t)$.\\
ii) For any $\phi\in L^2$ the uniform bounds hold:
\be\la{6.4}
\Vert U'_0(t)\phi\Vert =\Vert \phi\Vert,\quad
\Vert U'(t)\phi\Vert =\Vert \phi\Vert,\quad t\ge 0.
\ee
iii) Under assumptions {\bf E1}-{\bf E2}
for $U'(t)$ a bound of type (\re{full}) also holds:
\be\la{full1}
\Vert P_cU'(t)\psi\Vert_{L^2_{-\si}}
\le C(1+|t|)^{-3/2}\Vert \psi\Vert_{L^2_{\si}},\quad t\in\R
\ee
with  $\si>5/2$.
\end{cor}
Now we  formulate the scattering
theory  in the dual representation.

\begin{theorem}\la{t6.1}
Let conditions  {\bf E1--E2} and {\bf S0--S3} hold and $\si>5/2$.
Then there exist linear operators
$W,r(t): L^2_{\si}\to  L^2$
such that for $\phi\in L^2_\si$
\be\la{dst}
P_cU'(t)\phi=U'_0(t)W\phi+r(t)\phi, \,\,\,t\ge 0.
\ee
and the bounds hold
\beqn
\Vert r(t)\phi\Vert&\le&
C(1+t)^{-1/2}\Vert\phi\Vert_{L^2_\si},
\la{6.9}\\
E|\langle\psi_0,r(t)\phi\rangle|^2&\le&
C(1+t)^{-1}   \Vert\phi\Vert^2_{L^2_\si},\quad t>0.
\la{6.6}
\eeqn
\end{theorem}
\begin{proof}
 We apply the Cook method, \ci[Theorem XI.4]{RS3}.
Fix $\phi\in L^2_{\si}$ and define $W\phi$,  formally, as
\be\la{W-int}
W\phi=\lim_{t\to+\infty}U'_0(-t)P_cU'(t)\phi=
\phi+\int_{0}^{+\infty}\frac{d}{d\tau}
U'_0(-\tau)P_cU'(\tau)\phi  \,d\tau.
\ee
We have to prove the convergence of the last integral in the norm
of $L^2$. First, observe that
$$
  \frac{d}{d\tau} U'_0(\tau)\phi={\cal A}'_0 U'_0(\tau)\phi,\quad
  \frac{d}{d\tau} U'(\tau)\phi={\cal A}'U'(\tau)\phi,\,\,\tau\ge 0
$$
where
${\cal A}'_0$ and  ${\cal A}'$
are the generators to the groups $U'_0(\tau)$,  $U'(\tau)$,
respectively. Therefore,
\be\la{6.10}
  \frac{d}{d\tau} U'_0(-\tau)P_c U'(\tau)\phi
  =U'_0(-\tau) ({\cal A}'-{\cal A}'_0)P_cU'(\tau)\phi.
\ee
We have
${\cal A}'-{\cal A}'_0 =iV$.
Furthermore, {\bf E2}, (\ref{6.4}), (\ref{full1}) imply that
\beqn\la{6.11}
  \!\!\!\!\!\! \!\!\!\!\!\!
  \Vert U'_0(-\tau) ({\cal A}-{\cal A}_0)P_c U'(\tau)\phi\Vert
  \!\!&\le&\!\!
  C~\Vert({\cal A}-{\cal A}_0)P_c U'(\tau)\phi\Vert
  =C~\Vert VU'(\tau)\phi\Vert\\
  \nonumber
  \!\!&\le&\!\! C_1~\Vert U'(\tau)\phi\Vert_{L^2_{-\rho}}
  \le C_2 (1+\tau)^{-3/2}  \Vert\phi\Vert_{L^2_\si},
  \quad\tau\ge 0.
\eeqn
Hence, the convergence of the integral
in the right hand side of (\re{W-int}) follows.
\smallskip\\
Further, (\re{dst}) and (\re{W-int}) imply
$$
r(t)\phi=P_cU'(t)\phi-U'_0(t)W\phi
=-U'_0(t)\int_t^\infty\frac{d}{d\tau}U'_0(-\tau)P_cU'(\tau)\phi\,d\tau.
$$
Hence (\re{6.9}) follows by  (\ref{6.4}), (\ref{6.10}) and (\ref{6.11}).
\medskip\\
It remains to prove (\re{6.6}).
Applying the Shur lemma we obtain
\beqn\la{6.14}
  E|\langle\psi_0,r(t)\phi\rangle|^2&=&
  \langle q_0(x-y),
  r(t)\phi(x)\otimes r(t)\phi(y)\rangle
  \nonumber\\
  &\le&\Vert q_0\Vert_{L^1}\,\,\Vert r(t)\phi \Vert^2.
\eeqn
Hence,  (\ref{6.6})  follows by (\ref{qp}) with $p=1$ and (\re{6.9}).
\end{proof}

\subsection{Convergence to equilibrium distribution}
 Theorem \re{tA} can be derived from Propositions \re{p9.1}-\re{p9.2}
below by using the same arguments as in \ci[Theorem XII.5.2]{VF}.
\begin{pro}\la{p9.1}
The family of the measures $\{P_c^*\mu_t, t\in\R\}$
is weakly compact in $P_cH^{-\ve}_{-\si}$ for any $\ve>0$
and $\si>5/2+\delta$.
\end{pro}
\begin{pro}\la{p9.2}
For any $\phi\in{\cal D}$
\be\la{2.6'}
 \widehat {P_c^*\mu_t}(\phi)
 \equiv\int \exp(i\langle\psi,\phi\rangle)P_c^*\mu_t(d\psi)
 \rightarrow
\exp\{-\fr{1}{2}{\cal Q}_\infty(W\phi,W\phi)\},
 \,\,\,t\to\infty.
 \ee
\end{pro}
Proposition \re{p9.1} provides the existence of the limiting measures
of the family $P_c^*\mu_t$, and Proposition \re{p9.2} provides the
uniqueness of the limiting measure, and hence the convergence (\re{1.8'}).
We deduce these propositions
with the help of Theorem \re{t6.1}.
\medskip\\
{\bf Proof of Proposition \ref{p9.1}.}
First, we prove the bound
\be\la{pest}
\sup\limits_{t\ge0} E\Vert P_cU(t)\psi_0\Vert_{\cH}<\infty,
\ee
Representation (\re{split}) implies
\be\la{7.3}
 P_c U(t)\psi_0=P_c\chi(x,t)+P_c\phi(x,t),
\ee
where $\chi(x,t)= U_0(t)\psi_0$,
and $\phi(x,t)$ is the solution to (\re{vsp}).
Therefore,
\be\la{7.4}
  E\Vert P_cU(t)\psi_0\Vert_{\cH}\le
  E\Vert P_c\chi(t)\Vert_{\cH}
  +E\Vert P_c\phi(t)\Vert_{\cH}.
\ee
Bound  (\ref{bpp}) implies
\be\la{chi-est}
\sup\limits_{t\ge0}E\Vert\chi(t)\Vert_{\cH}
<\infty.
\ee
Further, we have by the Cauchy-Schwartz inequality
$$
E\Vert(\chi(t),\zeta_{j})\zeta_{j}\Vert_{L^2_{-\si}}
\le C\Vert\zeta_{j}\Vert_{L^2_{-\si}}
\Vert\zeta_{j}\Vert_{L^2_{\si}}
E\Vert\chi(t)\Vert_{L^2_{-\si}}
\le C_jE\Vert\chi(t)\Vert_{L^2_{-\si}},\quad\si=5/2+\de
$$
since the eigenfunctions $\zeta_j\in L^2_s$ with any  $s$,
see Appendix.
Therefore
$$
\sup\limits_{t\ge0} E\Vert P_c\chi(t)\Vert_{\cH}
<\infty
$$
since $P_c\chi(x,t)=\chi(x,t)-P_d\chi(x,t)$  by (\re{pc}).

It remains to estimate the second term in the RHS of (\ref{7.4}).
Choose a $\delta_1>0$ such that $\delta_1<\rho-5-\delta$.
It is possible due to {\bf E1}.
Then the  Duhamel representation  (\ref{vsp})
and  bounds (\re{full}) and (\re{chi-est}) imply
\beqn\nonumber
\!\!\!\!\!\!E\Vert P_c\phi(t)\Vert_{\cH}\!\!&\le&\!\!\int_0^t
E\Vert P_cU(t-s)V\chi(s)\Vert_{L^2_{-5/2-\de}}~ds
\le C\int_0^t (1+t-s)^{-3/2}E\Vert V\chi(t)\Vert_{L^2_{5/2+\delta_1}}ds\\
\la{phi-est}
\!\!&\le&\!\! C_1\int_0^t (1+t-s)^{-3/2}
E\Vert \chi(t)\Vert_{L^2_{5/2+\delta_1-\rho}}ds
\le C_2,\quad t>0
\eeqn
since $5/2+\delta_1-\rho<-5/2-\delta$.
Now  (\ref{7.4})-- (\ref{phi-est}) imply (\re{pest}).
\medskip\\
Now Proposition \re{p9.1} follows from (\re{pest}) by Prokhorov theorem
\ci[Lemma II.3.1]{VF} as in the proof of \ci[Theorem XII.5.2]{VF}.
\hfill$\Box$
 \medskip\\
{\bf Proof of Proposition \ref{p9.2}}
We have
$$
\int \exp(i\langle\psi,\phi\rangle)P_c^*\mu_t(d\psi)=
\int \exp(i\langle P_c\psi,\phi\rangle)\mu_t(d\psi)
=E\exp{i\langle P_cU(t)\psi_0,\phi\rangle}
$$
Bound (\ref{6.6}) and  Cauchy-Schwartz inequality imply
\beqn\nonumber
|E\exp{i\langle P_cU(t)\psi_0,\phi\rangle}-
E\exp{i\langle U_0(t)\psi_0,W\phi\rangle}|\!\!&=&\!\!
|E\exp{i\langle \psi_0,P_cU'(t)\phi\rangle}-
E\exp{i\langle\psi_0,U'_0(t)W\phi\rangle}|\\
\nonumber
\!\!&\le&\!\!E|\langle\psi_0,r(t)\phi\rangle|
\le(E\langle\psi_0,r(t)\phi\rangle^2)^{1/2}\to 0
\eeqn
as $t\to\infty$. It remains to prove that
\be\la{7.10}
E \exp{i\langle\psi_0,U'_0(t)W\phi\rangle}  \to
\exp\{-\fr{1}{2}{\cal  Q}_\infty( W\phi,W\phi)\},
~~t\to\infty.
\ee
The convergence does  not follow directly  from
Proposition \ref{p2} since
$ W\phi\not\in{\cal D}$.
We  can approximate $W\phi\in L^2$
by functions from ${\cal D}$
since ${\cal D}$ is dense in $L^2$.
Hence, for any $\ve>0$ there exists
$\phi_\ve\in{\cal D}$ such that
\be\la{7.12}
\Vert  W\phi-\phi_\ve \Vert\le\ve.
\ee
By the triangle
inequality
\beqn
&&|E \exp{i\langle\psi_0,U'_0(t)W\phi\rangle}-
\exp\{-\fr{1}{2}{\cal Q}_\infty (W\phi,W\phi)\}|\nonumber\\
&\le&
|E \exp{i\langle\psi_0,U'_0(t)W\phi\rangle}-
E\exp{i\langle\psi_0,U'_0(t)\phi_\ve\rangle}|\nonumber\\
&&
+E|\exp{i\langle U_0(t)\psi_0,\phi_\ve\rangle}-
\exp\{-\frac{1}{2} {\cal Q}_{\infty}(\phi_\ve,\phi_\ve)\}|
\nonumber\\
&&+|\exp\{-\frac{1}{2}{\cal Q}_{\infty}(\phi_\ve,\phi_\ve)\}-
\exp\{-\frac{1}{2}{\cal Q}_{\infty}(W\phi, W\phi)\}|.
\la{7.13}
\eeqn
Let us estimate each term in the RHS of (\ref{7.13}).
Theorem \ref{t6.1}
implies that
uniformly in $t>0$
\beqn\nonumber
E|\langle\psi_0,U'_0(t)(W\phi-\phi_\ve)\rangle|&\le&
(E|\langle\psi_0,U'_0(t)(W\phi-\phi_\ve)\rangle|^2)^{1/2}\le
\Vert q_0\Vert_{L^1}^{1/2}\Vert U'_0(t)(W\phi-\phi_\ve)\Vert\\
\nonumber
&\le& C\Vert W\phi-\phi_\ve \Vert\le C\ve.
\eeqn
Then the first term is ${\cal O}(\ve)$
 uniformly in $t>0$.
The second term
converges to zero as $t\to\infty$ by Proposition \ref{p2}
since  $\phi_\ve\in {\cal D}$.
At last, the third term
is ${\cal O}(\ve)$
by (\ref{7.12})  and  the continuity
of the quadratic form ${\cal Q}_\infty(\phi,\phi)$
in $L^2\otimes\C^4$. The continuity follows
from Corollary \ref{coro}.
Now  convergence (\ref{7.10}) follows
since $\ve>0$ is arbitrary.
\hfill$\Box$
\setcounter{equation}{0}
\section{Appendix: Decay of eigenfunctions}

Here we prove the spatial decay of eigenfunctions.
\begin{lemma}\la{ef-d}
Let $V$ satisfy {\bf E1}, and $\psi(x)\in L^2(\R^3)$ be an
eigenfunction of the Dirac operator corresponding to a eigenvalue
$\lam\in(-m,m)$, i.e.
$$
H\psi(x)=\lam\psi(x),\quad x\in\R^3.
$$
Then $\psi\in L^2_s$ for all $s\in\R$.
\end{lemma}
\begin{proof}
Denote by $R_0(\lam)=(H_0-\lam)^{-1}$ the resolvent of the free
Dirac equation.
The equation  $(H_0+V-\lam)\psi=0$ implies
\be\la{pR}
\psi=R_0(\lam)f,\quad {\rm where}\quad f=-V\psi\in L^2_{2+\rho}
\ee
From the identity
$$
(-i\al\cdot\nabla+\beta m-\lam)(i\al\cdot\nabla-\beta m-\lam)
=\Delta-m^2+\lam^2
$$
it follows that
\be\la{RR}
R_0(\lam)=\fr{i\al\cdot\nabla-\beta m-\lam}{\Delta-m^2+\lam^2}
\ee
Hence, in the Fourier transform, the first equation of (\re{pR}) reads
$$
\hat\psi(k)=\frac{(-\al\cdot k+\beta m+\lam)\hat f(k)}
{k^2+m^2-\lam^2}
$$
Since $|\lam|<m$, we have
$$
\Vert\psi\Vert_{L^2_{2+\rho}}=C\Vert\hat\psi\Vert_{H^{2+\rho}}
\le C_1\Vert\hat f\Vert_{H^{2+\rho}}=
C_2\Vert f\Vert_{L^2_{2+\rho}}
\le C_3\Vert\psi\Vert_{L^2_{2}}
$$
Hence, $\psi\in L^2_s$ with any $s\in\R$ by induction.
\end{proof}

\end{document}